\documentclass[11pt]{article} 
\usepackage[margin=1in]{geometry}
\usepackage[utf8]{inputenc}
\usepackage[T1]{fontenc}
\usepackage[english]{babel}
\usepackage{lmodern} 
\usepackage{graphicx}
\usepackage{wrapfig}
\usepackage{subfig}
\usepackage{caption}
\usepackage{paralist}

\graphicspath{{./figures/}}

\newcommand{\titel}{Mondshein Sequences (a.k.a.\ $(2,1)$-Orders)}

\usepackage{algorithm} 
\usepackage[noend]{algpseudocode} 

\usepackage[usenames]{color} 
\definecolor{hellblau}{rgb}{0.2,0.4,1} 
\definecolor{dunkelblau}{rgb}{0,0,0.8}
\definecolor{dunkelgruen}{rgb}{0,0.5,0}
\usepackage[
	pdftex,
	colorlinks,
	linkcolor=dunkelblau,
	urlcolor=dunkelblau,
	citecolor=dunkelgruen,
	bookmarks=true,
	linktocpage=true,
	pdfauthor={Jens M. Schmidt},
	pdfsubject={}
]{hyperref} 
\usepackage{pdfpages}

\usepackage{amsmath} 
\usepackage{amsthm} 
\usepackage{amsfonts} 
\theoremstyle{plain} 
	\newtheorem{satz}{Satz}[] 
	\newtheorem{theorem}[satz]{Theorem}
	\newtheorem{lemma}[satz]{Lemma}
	\newtheorem{observation}[satz]{Observation}
	
\theoremstyle{remark} 
	\newtheorem*{remark}{\textbf{Remark}} 
\theoremstyle{definition} 
	\newtheorem{definition}[satz]{Definition}
	
	\newtheorem*{conjecture}{Conjecture}

\begin{document}
	\title{\titel}
		\author{Jens M. Schmidt\\
		Institute of Mathematics\\
		TU Ilmenau\footnote{This research was partly done at Max Planck Institute for Informatics, Saarbrücken.
		An extended abstract of this paper has been published at ICALP'14.}}
	\date{}
	\maketitle

\begin{abstract}
Canonical orderings [STOC'88, FOCS'92] have been used as a key tool in graph drawing, graph encoding and visibility representations for the last decades. We study a far-reaching generalization of canonical orderings to non-planar graphs that was published by Lee Mondshein in a PhD-thesis at M.I.T.\ as early as 1971.

Mondshein proposed to order the vertices of a graph in a sequence such that, for any $i$, the vertices from $1$ to $i$ induce essentially a $2$-connected graph while the remaining vertices from $i+1$ to $n$ induce a connected graph. Mondshein's sequence generalizes canonical orderings and became later and independently known under the name \emph{non-separating ear decomposition}. Surprisingly, this fundamental link between canonical orderings and non-separating ear decomposition has not been established before. Currently, the fastest known algorithm for computing a Mondshein sequence achieves a running time of $O(nm)$; the main open problem in Mondshein's and follow-up work is to improve this running time to subquadratic time.

After putting Mondshein's work into context, we present an algorithm that computes a Mondshein sequence in optimal time and space $O(m)$. This improves the previous best running time by a factor of $n$. We illustrate the impact of this result by deducing linear-time algorithms for five other problems, for four out of which the previous best running times have been quadratic. In particular, we show how to
\begin{itemize}
	\item[--] compute three independent spanning trees in a $3$-connected graph in time $O(m)$, improving a result of Cheriyan and Maheshwari [J. Algorithms 9(4)],
	\item[--] improve the preprocessing time from $O(n^2)$ to $O(m)$ for the output-sensitive data structure by Di Battista, Tamassia and Vismara [Algorithmica 23(4)] that reports three internally disjoint paths between any given vertex pair,
	\item[--] derive a very simple $O(n)$-time planarity test once a Mondshein sequence has been computed,
	\item[--] compute a nested family of contractible subgraphs of $3$-connected graphs in time $O(m)$,
	\item[--] compute a $3$-partition in time $O(m)$, while the previous best running time is $O(n^2)$ due to Suzuki et al.\ [IPSJ 31(5)].
\end{itemize}
\end{abstract}

\section{Introduction}
Canonical orderings are a fundamental tool used in graph drawing, graph encoding and visibility representations; we refer to~\cite{Badent2011} for a wealth of applications. For maximal planar graphs, canonical orderings were introduced by de Fraysseix, Pach and Pollack~\cite{Fraysseix1988,Fraysseix1990} in 1988. Kant then generalized canonical orderings to $3$-connected planar graphs~\cite{Kant1992,Kant1996}. In polyhedral combinatorics, canonical orders are in addition related to shellings of (dual) convex 3-dimensional polytopes~\cite{Ziegler1998}; however, such shellings are often, as in the Bruggesser-Mani theorem, dependent on the geometry of the polytope. A combinatorial generalization to arbitrary planar graphs was given by Chiang, Lin and Lu~\cite{Chiang2005}.

Surprisingly, the concept of canonical orderings can be traced back much further, namely to a long-forgotten PhD-thesis at {M.I.T.}\ by Lee F. Mondshein~\cite{Mondshein1971} in 1971. In fact, Mondshein proposed a sequence that generalizes canonical orderings to non-planar graphs, hence making them applicable to arbitrary $3$-connected graphs. \emph{Mondshein's sequence} was, independently and in a different notation, found later by Cheriyan and Maheshwari~\cite{Cheriyan1988} under the name \emph{non-separating ear decompositions} and is sometimes also called \emph{(2,1)-order} (e.g., see~\cite{Biedl2016}). In addition, Mondshein sequences provide a generalization of Schnyde{r's} famous woods to non-planar 3-connected graphs. One key contribution of this paper is to establish the above fundamental link between canonical orderings and non-separating ear decompositions in detail.

Computationally, it is an intriguing question how fast a Mondshein sequence can be computed. Mondshein himself gave an involved algorithm with running time $O(m^2)$. Cheriyan showed that it is possible to achieve a running time of $O(nm)$ by using a theorem of Tutte that proves the existence of non-separating cycles in $3$-connected graphs~\cite{Tutte1963}. Both works state as main open problem, whether it is possible to compute a Mondshein sequence in subquadratic time (see~\cite[p.\ 1.2]{Mondshein1971} and~\cite[p.\ 532]{Cheriyan1988}).

We present the first algorithm that computes a Mondshein sequence in optimal time and space $O(m)$, hence solving the above 40-year-old problem. The interest in such a computational result stems from the fact that $3$-connected graphs play a crucial role in algorithmic graph theory. We illustrate this in five applications by giving linear-time algorithms. For four of them, the previous best running times have been quadratic.

We start by giving an overview of Mondshein's work and its connection to canonical orderings and non-separating ear decompositions in Section~\ref{sectiongeneralizing}. Section~\ref{sec:computing} explains the linear-time algorithm and proves its main technical lemma, the Path Replacement Lemma. Section~\ref{sec:applications} covers five applications of our linear-time algorithm.

\section{Preliminaries}
We use standard graph-theoretic terminology and assume that all graphs are simple.

\begin{definition}[\cite{Lovasz1985,Whitney1932a}]
An \emph{ear decomposition} of a graph $G=(V,E)$ is a sequence $(P_0,P_1,\ldots,P_k)$ of subgraphs of $G$ that partition $E$ such that $P_0$ is a cycle and every $P_i$, $1 \leq i \leq k$, is a path that intersects $P_0 \cup \cdots \cup P_{i-1}$ in exactly its endpoints. Each $P_i$ is called an \emph{ear}. An ear is \emph{short} if it is an edge and \emph{long} otherwise.
\end{definition}

According to Whitney~\cite{Whitney1932a}, every ear decomposition has exactly $m-n+1$ ears and $G$ has an ear decomposition if and only if $G$ is 2-connected. For any $i$, let $G_i := P_0 \cup \cdots \cup P_i$ and $\overline{V_i} := V - V(G_i)$. We write $\overline{G_i}$ to denote the graph induced by $\overline{V_i}$. Note that $\overline{G_i}$ does not necessarily contain all edges in $E-E(G_i)$; in particular, there may be short ears in $E-E(G_i)$ that have both endpoints in $G_i$.

For a path $P$ and two vertices $x$ and $y$ in $P$, let \emph{$P[x,y]$} be the subpath in $P$ from $x$ to $y$. A path with endpoints $v$ and $w$ is called a \emph{$vw$-path}. A vertex $x$ in a $vw$-path $P$ is an \emph{inner vertex} of $P$ if $x \notin \{v,w\}$. For convenience, every vertex in a cycle is called an \emph{inner vertex} of that cycle.

For an ear $P$, let $inner(P)$ the set of its inner vertices. The inner vertex sets of the ears in an ear decomposition of $G$ play a special role, as they partition $V$. Every vertex of $G$ is contained in exactly one long ear as inner vertex. This gives readily the following characterization of $\overline{V_i}$.

\begin{observation}\label{observation}
For every $i$, $\overline{V_i}$ is the union of the inner vertices of all long ears $P_j$ with $j > i$.
\end{observation}

We will compare vertices and edges of $G$ by their first occurrence in a fixed ear decomposition.

\begin{definition}
Let $D = (P_0,P_1,\ldots,P_{m-n})$ be an ear decomposition of $G$. For an edge $e \in G$, let $birth_D(e)$ be the index $i$ such that $P_i$ contains $e$. For a vertex $v \in G$, let $birth_D(v)$ be the minimal $i$ such that $P_i$ contains $v$ (thus, $P_{birth_D(v)}$ is the ear containing $v$ as an inner vertex). Whenever $D$ is clear from the context, we will omit $D$.
\end{definition}

Clearly, for every vertex $v$, the ear $P_{birth(v)}$ is long, as it contains $v$ as an inner vertex.

\section{Generalizing Canonical Orderings}\label{sectiongeneralizing}
Although canonical orderings of (maximal or 3-connected) planar graphs are traditionally defined as vertex partitions, we will define them as special ear decompositions. This will allow for an easy comparison of canonical orderings to the more general Mondshein sequences, which extend them to non-planar graphs. We assume that the input graphs are $3$-connected and, when talking about canonical orderings, planar. It is well-known that maximal planar graphs (which were considered in~\cite{Fraysseix1988} in this setting) form a subclass of $3$-connected graphs, apart from the triangle-graph.

\begin{definition}\label{def:nonseparating}
An ear decomposition is \emph{non-separating} if, for every long ear $P_i$ except the last one, every inner vertex of $P_i$ has a neighbor in $\overline{G_i}$.
\end{definition}

The name \emph{non-separating} refers to the following helpful property.

\begin{lemma}\label{lem:connected}
In a non-separating ear decomposition $D$, $\overline{G_i}$ is connected for every $i$.
\end{lemma}
\begin{proof}
For all $i$ satisfying $\overline{G_i} = \emptyset$ the claim is true, in particular if $i$ is at least the index of the last long ear. Otherwise, $i$ is such that the inner vertex set $A$ of the last long ear in $D$ is contained in $\overline{G_i}$. Consider any vertex $x$ in $\overline{G_i}$. In order to show connectedness, we exhibit a path from $x$ to $A$ in $\overline{G_i}$. If $x \in A$, we just take the path of length zero. Otherwise, the vertex $x$ has a neighbor in $\overline{G_{birth(x)}}$, since $D$ is non-separating. According to Observation~\ref{observation}, this neighbor is an inner vertex of some ear $P_j$ with $j > birth(x)$. Applying induction on $j$ gives the desired path to $A$.
\end{proof}

A \emph{plane graph} is a graph that is embedded into the plane. In particular, a plane graph has a fixed outer face. We define canonical orderings as follows.

\begin{definition}[canonical ordering]\label{def:canonicalordering}
Let $G$ be a $3$-connected plane graph and let $rt$ and $ru$ be edges of its outer face. A \emph{canonical ordering} through $rt$ and avoiding $u$ is an ear decomposition $D$ of $G$ such that
\begin{compactenum}
	\item[1.] $rt \in P_0$,
	\item[2.] $P_{birth(u)}$ is the last long ear, contains $u$ as its only inner vertex and does not contain $ru$, and
	\item[3.] $D$ is non-separating.
\end{compactenum}
\end{definition}

The fact that $D$ is non-separating plays a key role for both canonical orderings and their generalization to non-planar graphs. E.g., Lemma~\ref{lem:connected} implies that the plane graph $G$ can be constructed from $P_0$ by successively inserting the ears of $D$ to only one dedicated face of the current embedding, a routine that is heavily applied in graph drawing and embedding problems. Put simply, the second condition forces $u$ to be ``added last'' in $D$. Further motivations are given by 3-connectivity: If we would not restrict $u$ to be the only vertex in $P_{birth(u)}$, other vertices in the same ear could have degree two, as the non-separateness does not imply any later neighbors for the last ear. The condition $ru \notin P_{birth(u)}$ ensures that $u$ has degree at least three in $G$ (which is necessary for $3$-connectivity) and will also lead to the existence of a third independent spanning tree (see Application~1 in Section~\ref{sec:applications}).

We note that forcing one edge $rt$ in $P_0$ is optimal in the sense that two edges $rz$ and $rt$ cannot be forced: Let $W$ be a sufficiently large wheel graph with center vertex $r$ and rim vertices $t$ and $z$ such that $t$ and $z$ are not adjacent. Then a canonical ordering with $rt, rz \in P_0$ and avoiding $u$ does not exist, as any inner vertex on the rim-path from $t$ to $z$ not containing $u$ has no larger neighbor with respect to $birth$, and thus violates the non-separateness.

The original definition of canonical orderings by Kant~\cite{Kant1996} states the following additional properties.

\begin{lemma}[further properties]\label{lem:canonical}
For every $0 \leq i \leq m-n$ in a canonical ordering,
\begin{compactenum}
	\item[4.] the outer face $C_i$ of the plane subgraph $G_i \subseteq G$ is a (simple) cycle that contains $rt$,
	\item[5.] $G_i$ is $2$-connected and every separation pair of $G_i$ has both its vertices in $C_i$, and
	\item[6.] for $i > 0$, the neighbors of $inner(P_i)$ in $G_{i-1}$ are contained consecutively in $C_{i-1}$.
\end{compactenum}
Further, the canonical ordering implies the existence of one satisfying the following property:
\begin{compactenum}
	\item[7.] if $|inner(P_i)| \geq 2$, each inner vertex of $P_i$ has degree two in $G-\overline{V_i}$
\end{compactenum}
\end{lemma}

Properties~4--6 can be easily deduced from Definition~\ref{def:canonicalordering} as follows: Every $G_i$ is a $2$-connected plane subgraph of $G$, as $G_i$ has an ear decomposition. According to~\cite[Corollary~1.3]{Thomassen1981}, all faces of a $2$-connected plane graph form cycles. Thus, every $C_i$ is a cycle and Property~4 follows directly from the fact that $rt$ is assumed to be in the fixed outer face of $G$. Property~5 is implied by the $3$-connectivity of $G$ and Property~4.
Property~6 follows from Property~4, the fact that every inner vertex of $P_i$ must be outside $C_{i-1}$ (in $G$) and the Jordan Curve Theorem.

For the sake of completeness, we show how Property~7 is derived. Although it is not directly implied by Definition~\ref{def:canonicalordering} (in that sense our definition is more general), the following lemma shows that we can always find a canonical ordering satisfying it.

\begin{lemma}\label{TransformingToInduced}
Every canonical ordering can be transformed to a canonical ordering satisfying Property~\ref{lem:canonical}.7 in linear time.
\end{lemma}
\begin{proof}
First, consider any ear $P_i \neq P_0$ with $|inner(P_i)| \geq 2$ such that an inner vertex $x$ of $P_i$ has a neighbor $y$ in $G-\overline{V_i}$ that is different from its predecessor and successor in $P_i$. Then $P_{birth(xy)}=xy$ and $birth(xy) > i$. If $y$ is in $P_i$, let $Z$ be the path obtained from $P_i$ by replacing $P_i[x,y] \subseteq P_i$ with $xy$; we call this latter operation \emph{short-cutting}. We replace $P_i$ with the two ears $Z$ and $P_i[x,y]$ in that order and delete $P_{birth(xy)}=xy$. This preserves Properties~1--3 (note that $u \notin P_i$, as $|inner(P_i)| \geq 2$) and therefore the canonical ordering. If $y$ is not in $P_i$, let $Z_1$ be a shortest path in $P_i$ from an endpoint of $P_i$ to $x$ and let $Z_2$ be the path in $P_i$ from $x$ to the remaining endpoint. Replace $P_i$ with the two ears $Z_1 \cup xy$ and $Z_2$ in that order and delete $P_{birth(xy)}$. This preserves Properties~1--3.

Now, consider a vertex $x \in P_0$ having not degree $2$ in $G-\overline{V_0}$, i.e.\ $x$ has a non-consecutive neighbor $y$ in $P_0$ in the graph vertex-induced by $V(P_0)$. If $x \in \{r,t\}$, we replace $P_0$ with the shortest cycle $C$ in $P_0 \cup xy$ that contains $r$, $t$ and $y$, delete $P_{birth(xy)}=xy$ and add the remaining path from $x$ to $y$ in $P_0 -E(C)$ as new ear directly after $C$. This clearly preserves Properties~1--3. If $x \notin \{r,t\}$, we can shortcut $P_0$ in a similar way. The above operations can be computed in linear total time.
\end{proof}

Our definition of canonical orderings uses planarity only in one place: $tr \cup ru$ is assumed to be part of the outer face of $G$. Note that the essential part of this assumption is that $tr \cup ru$ is part of \emph{some} face of $G$, as we can always choose an embedding for $G$ having this face as outer face. Hence, there is a natural generalization of canonical orderings to non-planar graphs $G$: We merely require $rt$ and $ru$ to be edges of $G$! The following ear-based definition is similar to the one given in~\cite{Cheriyan1988} but does not need additional degree-constraints.

\begin{definition}[\cite{Mondshein1971,Cheriyan1988}]\label{def:Mondshein}
Let $G$ be a graph with edges $rt$ and $ru$. A \emph{Mondshein sequence through} $rt$ and \emph{avoiding} $u$ (see Figure~\ref{fig:SequenceCr08Paper}) is an ear decomposition $D$ of $G$ such that
\begin{compactenum}
	\item[1.] $rt \in P_0$,
	\item[2.] $P_{birth(u)}$ is the last long ear, contains $u$ as its only inner vertex and does not contain $ru$, and
	\item[3.] $D$ is non-separating.
\end{compactenum}
\end{definition}

This definition is in fact equivalent to the one Mondshein used 1971 to define a \emph{(2,1)-sequence} \cite[Def.\ 2.2.1]{Mondshein1971}, but which he gave in the notation of a special vertex ordering. This vertex ordering actually refines the partial order $inner(P_0),\ldots,inner(P_{m-n})$ by enforcing an order on the inner vertices of each path according to their occurrence on that path (in any direction). The statement that canonical orderings can be extended to non-planar graphs can also be found in~\cite[p.113]{Fraysseix1994}, however, no further explanation is given.

\begin{figure}[h!tb]
	\centering
	\includegraphics[scale=0.7]{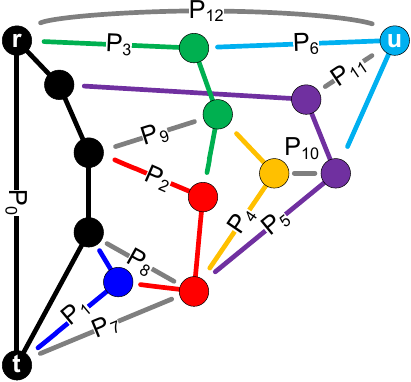}
	\caption{A Mondshein sequence of a non-planar $3$-connected graph.}
	\label{fig:SequenceCr08Paper}
\end{figure}

Note that Definition~\ref{def:Mondshein} implies $u \notin P_0$, as $P_0 \neq P_{birth(u)}$, since $P_{birth(u)}$ contains only one inner vertex. As a direct consequence of this and the fact that $D$ is non-separating, $G$ must have minimum degree at least $3$ in order to have a Mondshein sequence. Mondshein proved that every $3$-connected graph has a Mondshein sequence. In fact, also the converse is true.

\begin{theorem}\cite{Cheriyan1988,Zehavi1989}\label{CheriyanThm}
Let $rt$ and $ru$ be edges of $G$. Then $G$ is $3$-connected if and only if $G$ has a Mondshein sequence through $rt$ and avoiding $u$.
\end{theorem}

We state two additional facts about Mondshein sequences. For the first, let $G$ be planar. Clearly, every canonical ordering of an embedding of $G$ is also a Mondshein sequence. Conversely, let $D$ be a Mondshein sequence of $G$ through $rt$ and avoiding $u$. Then Theorem~\ref{CheriyanThm} implies that $G$ is 3-connected. If $G$ has an embedding in which $tr \cup ru$ is contained in a face, we can choose this face as outer face and get an embedding of $G$ for which $D$ is a canonical ordering. This embedding must be unique, as Whitney proved that any 3-connected planar graph has a unique embedding (up to flipping)~\cite{Whitney1932}. Otherwise, there is no embedding of $G$ such that $tr \cup ru$ is contained in some face. Since the faces of a $3$-connected planar graph are precisely its non-separating cycles~\cite{Tutte1963}, we conclude the following observation.

\begin{observation}\label{obs:MondsheinIsCanonical}
For a planar graph $G$ and edges $tr$ and $ru$, the following statements are equivalent:
\begin{compactitem}
	\item There is a planar embedding of $G$ whose outer face contains $tr \cup ru$, and $D$ is a canonical ordering of this (unique) embedding through $rt$ and avoiding $u$.
	\item $D$ is a Mondshein sequence through $rt$ and avoiding $u$, and $tr \cup ru$ is contained in a non-separating cycle of $G$.
\end{compactitem}
\end{observation}

For the second fact, let a \emph{chord} of an ear $P_i$ be an edge in $G$ that joins two non-adjacent vertices of $P_i$. Note that the definition of a Mondshein sequence allows chords for every $P_i$. Once having a Mondshein sequence, one can aim for a slightly stronger structure. Let a Mondshein sequence be \emph{induced} if $P_0$ is induced in $G$ and every ear $P_i \neq P_0$ has no chord, except possibly the one joining the endpoints of $P_i$. It has been shown~\cite{Cheriyan1988} that every Mondshein sequence can be made induced. The following lemma shows the somewhat stronger statement that we can always expect Mondshein sequences to satisfy Property~\ref{lem:canonical}.7. In fact, its proof is precisely the same as the one for Lemma~\ref{TransformingToInduced}, since none of its arguments uses planarity.

\begin{lemma}
Every Mondshein sequence can be transformed to a Mondshein sequence $D$ satisfying Property~\ref{lem:canonical}.7 in linear time. In particular, $D$ is induced.
\end{lemma}

\section{Computing a Mondshein Sequence}\label{sec:computing}
Mondshein gave an involved algorithm~\cite{Mondshein1971} that computes his sequence in time $O(m^2)$. Independently, Cheriyan and Maheshwari gave an algorithm that runs in time $O(nm)$ and which is based on a theorem of Tutte. At the heart of our linear-time algorithm is the following classical construction sequence for $3$-connected graphs due to Barnette and Grünbaum~\cite{Barnette1969} and Tutte~\cite[Thms.\ 12.64 and~12.65]{Tutte1966}.

\begin{definition}\label{def:bgoperation}
The following operations on simple graphs are \emph{BG-operations} (see Figure~\ref{fig:BGOperations}).
\begin{compactenum}[(a)]
	\item \emph{vertex-vertex-addition}: Add an edge between two distinct non-adjacent vertices
	\item \emph{edge-vertex-addition}: Subdivide an edge $ab$, $a \neq b$, with a vertex $v$ and add the edge $vw$ for a vertex $w \notin \{a,b\}$
	\item \emph{edge-edge-addition}: Subdivide two distinct edges (the edges may intersect in one vertex) with vertices $v$ and $w$, respectively, and add the edge $vw$
\end{compactenum}
\end{definition}

\begin{figure}[h!tb]
	\centering
	\subfloat[vertex-vertex-addition]{
	\makebox[3.7cm]{
		\includegraphics[scale=0.7]{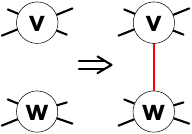}
		\label{fig:BG1}
	}}
	\hfill
	\subfloat[edge-vertex-addition]{
		\includegraphics[scale=0.7]{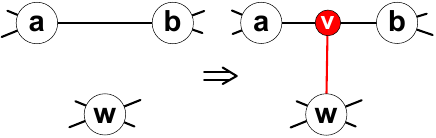}
		\label{fig:BG2}
	}
	\hfill
	\subfloat[edge-edge-addition]{
		\includegraphics[scale=0.7]{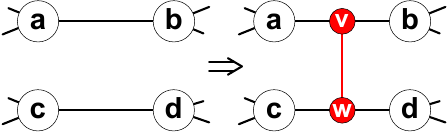}
		\label{fig:BG3}
	}
	\caption{BG-operations}
	\label{fig:BGOperations}
\end{figure}

\begin{theorem}[\cite{Barnette1969,Tutte1966}]\label{thm:BG}
A graph is 3-connected if and only if it can be constructed from $K_4$ using BG-operations.
\end{theorem}

Hence, applying a BG-operation on a $3$-connected graph preserves it to be simple and $3$-connected. Let a \emph{BG-sequence} of a $3$-connected graph $G$ be a sequence of BG-operations that constructs $G$ from $K_4$. It has been shown that such a BG-sequence can be computed efficiently.

\begin{theorem}[{\cite[Thms. 6.(2) and 52]{Schmidt2013}}]\label{thm:schmidt}
A BG-sequence of a $3$-connected graph can be computed in time $O(m)$.
\end{theorem}

The outline of our algorithm is as follows. Assume we want a Mondshein sequence of $G$ through $r\overline{t}$ and avoiding $\overline{u}$. We will first compute a suitable BG-sequence of $G$ using Theorem~\ref{thm:schmidt} and start with a Mondshein sequence of its first graph, the $K_4$. The crucial part is then a careful analysis that a Mondshein sequence of a $3$-connected graph can be modified to one of $G'$, where $G'$ is obtained from the former by applying a BG-operation.

In more detail, we need a special BG-sequence to harness the dynamics of the vertices $r$, $\overline{t}$ and $\overline{u}$ throughout the BG-sequence. A BG-sequence is determined by an (arbitrary) DFS-tree and two fixed incident edges of its root. We choose a DFS-tree with root $r$ and fix the edges $r\overline{t}$ and $r\overline{u}$. This way the initial $K_4$ will contain the vertex $r$ and $r$ will never be relabeled~\cite[Section~5]{Schmidt2010}.

However, $\overline{t}$ and $\overline{u}$ are not necessarily vertices of the $K_4$. This is a problem, as we have to specify an edge $rt$ and vertex $u$ of $K_4$ which the Mondshein sequence of $K_4$ goes through and avoids, respectively, for induction purposes. Fortunately, the relation between the graphs in a BG-sequence and subdivisions of these graphs in $G$~\cite[Section~4]{Schmidt2010} gives us such replacement vertices for $\overline{t}$ and $\overline{u}$ efficiently: We find vertices $t$ and $u$ of the initial $K_4$ such that the following labeling process ends with the input graph $G$ in which $t = \overline{t}$ and $u = \overline{u}$: For every BG-operation of the BG-sequence from $K_4$ to $G$ that subdivides the edge $rt$ or $ru$, we label the subdividing vertex with $t$ or $u$, respectively (the old vertex $t$ or $u$ is then given a different label). As desired, the final $t$ and $u$ upon completion of the BG-sequence will be $\overline{t}$ and $\overline{u}$. We refer to~\cite[Section~4]{Schmidt2010} for details on how to efficiently compute such a labeling scheme.

For the $K_4$, it is easy to compute a Mondshein sequence through $rt$ and avoiding $u$ efficiently. We iteratively proceed to a Mondshein sequence of the next graph in the sequence. The following modifications and their computational analysis are the main technical contribution of this paper and depend on the various positions in the sequence in which the vertices and edges that are involved in the BG-operation can occur.

Note that any short ear $xy$ in a Mondshein sequence can be moved to an arbitrary position of the sequence without destroying the Mondshein property, as long as both $x$ and $y$ are created at an earlier position. Thus, the essential information of a Mondshein sequence is its order on long ears. We will prove that there is always a modification that is local in the sense that the only long ears that are modified are the ones containing a vertex that is involved in the BG-operation.

\begin{lemma}[Path Replacement Lemma]\label{lem:PathReplacement}
Let $G$ be a $3$-connected graph with edges $rt$ and $ru$ and let $D=(P_0,P_1,\ldots,P_{m-n})$ be a Mondshein sequence of $G$ through $rt$ and avoiding $u$. Let $G'$ be obtained from $G$ by applying a BG-operation $\Gamma$ and let $rt'$ and $ru'$ be the edges of $G'$ that correspond to $rt$ and $ru$ in $G$. Then a Mondshein sequence $D'$ of $G'$ through $rt'$ and avoiding $u'$ can be computed from $D$ using only constantly many (amortized) constant-time modifications.
\end{lemma}

We split the proof into three parts. First, we state two preprocessing routines $leg()$ and $belly()$ on $D$ that will reduce the number of subsequent cases considerably. Second, we show how to modify $D$ to $D'$ using these routines and, third, we discuss computational issues.

From now on, let $vw$ be the edge that was added by $\Gamma$ such that $v$ subdivides $ab \in E(G)$ and $w$ subdivides $cd \in E(G)$ (if applicable). Thus, the vertex $u'$ in $G'$ is either $u$, $v$ or $w$, and likewise $t'$ in $G'$ is either $t$, $v$ or $w$.
By symmetry, we assume w.l.o.g.\ that $birth(a) \leq birth(b)$, $birth(c) \leq birth(d)$ and $birth(d) \leq birth(b)$. Recall that $\{a,b\}$ may intersect $\{c,d\}$ in at most one vertex. If not stated otherwise, the $birth$-operator refers always to $D$ in this section.

We need some notation for describing the modifications. Suppose $P_i$ is an ear containing an inner vertex $z$. If an orientation of $P_i$ is given, let $P_i[,z]$ be the prefix of $P_i$ ending at $z$ in this orientation and let $P_i[z,]$ be the suffix of $P_i$ starting at $z$. Occasionally, the orientation does not matter; if none is given, an arbitrary orientation can be taken. For paths $A$ and $B$ that end and start at a unique common vertex, let $A+B$ be the \emph{concatenation} of $A$ and $B$. Similarly, for disjoint paths $A$ and $B$ such that exactly one endpoint $x$ of $A$ is a neighbor of exactly one endpoint $y$ of $B$, let $A+B$ be the path $A \cup xy \cup B$.

\paragraph{Of legs and bellies:}
We describe two preprocessing routines. These will be used on $D$ in the next section to ensure that $ab \in P_{birth(b)}$ and $cd \in P_{birth(d)}$ (up to some special cases). Let an edge $xy \notin P_{birth(y)}$ be a \emph{leg} of $P_{birth(y)}$ if $xy \neq ru$ and $birth(x) < birth(y)$. For each such leg, $P_{birth(y)}$ is a long ear, $xy$ is a short ear, and $x$ is either not contained in $P_{birth(y)}$ or an endpoint of $P_{birth(y)}$ (see Figures~\ref{fig:Leg} and~\ref{fig:Leg2}). In the first case, if $y$ is not the only inner vertex of $P_{birth(y)}$, orient $P_{birth(y)}$ such that the successor of $y$ is also an inner vertex of $P_{birth(y)}$; this will preserve the non-separateness at $y$ for some later cases. In the latter case, orient $P_{birth(y)}$ toward $x$.

\begin{figure}[htb]
	\centering
	\subfloat[A leg $xy$ with $x \notin P_{birth(y)}$ and the result of Operation $leg(x,y)$ (dashed lines).]{
	\makebox[5cm]{
		\includegraphics[scale=0.8]{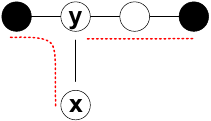}
		\label{fig:Leg}
	}
	}
	\hspace{1cm}
	\subfloat[A leg $xy$ with $x \in P_{birth(y)}$ and the result of Operation $leg(x,y)$.]{
	\makebox[5cm]{
		\includegraphics[scale=0.8]{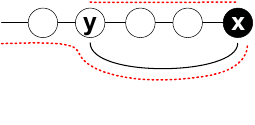}
		\label{fig:Leg2}
	}
	}
	\\
	\subfloat[A belly $xy$ with $birth(y) > 0$ and the result of Operation $belly(x,y)$.]{
	\makebox[5cm]{
		\includegraphics[scale=0.8]{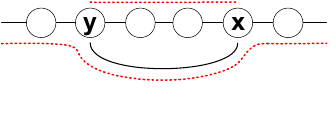}
		\label{fig:Belly}
	}
	}
	\hspace{1cm}
	\subfloat[A belly $xy$ with $birth(y) = 0$  and the result of Operation $belly(x,y)$.]{
	\makebox[5cm]{
		\includegraphics[scale=0.8]{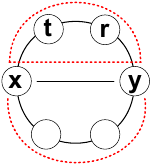}
		\label{fig:Belly2}
	}
	}
	\caption{}
\end{figure}

A leg $xy$ of $P_{birth(y)}$ has the feature that it may be incorporated into $P_{birth(y)}$ such that the resulting sequence is still a Mondshein sequence: Let $leg(x,y)$ be the operation that deletes the short ear $xy$ in the sequence $D$ and replaces the long ear $P_{birth(y)}$ by the two ears $P_{birth(y)}[,y]+x$ and $P_{birth(y)}[y,]$ in that order. We prove that the resulting sequence $\overline{D}$ is a Mondshein sequence. Clearly, $\overline{D}$ is an ear decomposition. In addition, we still have $rt \in P_0$, as $P_0$ did not change due to $birth(y) > birth(x) \geq 0$. Since every inner vertex of the two new ears is also an inner vertex of $P_{birth(y)}$, it has a neighbor in some larger ear (with respect to $birth$) in $\overline{D}$; thus $\overline{D}$ is non-separating by Definition~\ref{def:nonseparating}. Since $xy \neq ru$, the last long ear in $\overline{D}$ does not contain $ru$. The last long ear in $\overline{D}$ may be different from the one in $D$ if $y = u$, but since the replacement does not introduce any new inner vertex, it will still contain the same vertex $u$ as only inner vertex. Hence, $\overline{D}$ is a Mondshein sequence through $rt$ and avoiding $u$ by Definition~\ref{def:Mondshein}.

Let an edge $xy$ of $G$ be a \emph{belly} of $P_{birth(y)}$ if $birth(x) = birth(y) \neq birth(xy)$. Then $P_{birth(y)}$ contains both $x$ and $y$ as inner vertices, but does not contain $xy$; hence $xy$ is a short ear (see Figures~\ref{fig:Belly} and~\ref{fig:Belly2}).

For a belly $xy$, we can again find a Mondshein sequence that ensures $xy \in P_{birth(y)}$. First, consider the case $birth(y) > 0$, in which we orient $P_{birth(y)}$ from $y$ to $x$. For this case, let $belly(x,y)$ be the operation that deletes the short ear $xy$ in the sequence $D$ and replaces the long ear $P_{birth(y)}$ by the two long ears $P_{birth(y)}[,y]+P_{birth(y)}[x,]$ and $P_{birth(y)}[y,x]$ in that order (see Figure~\ref{fig:Belly}). For the same reasons as before, the resulting sequence $\overline{D}$ is an ear decomposition and non-separating. Since $P_{birth(y)}$ contains two inner vertices, we have $birth(y) \neq birth(u)$, and it follows that the last long ear in $\overline{D}$ is exactly the last long ear of $D$. In addition, $rt \in P_0$, as $P_0$ did not change due to $birth(y) > birth(x) \geq 0$. Hence, $\overline{D}$ is a Mondshein sequence through $rt$ and avoiding $u$.

Now consider the case $birth(y) = 0$. The vertices $x$ and $y$ cut $P_0$ into two distinct paths $A$ and $B$ having endpoints $x$ and $y$; let $A$ be the one containing $rt$. Let $belly(x,y)$ be the operation that deletes the short ear $xy$ in $D$ and replaces $P_0$ by the two long ears $A \cup xy$ and $B$ in that order (see Figure~\ref{fig:Belly2}). This preserves $P_0$ to be a cycle that contains $rt$ and, thus, gives also a Mondshein sequence through $rt$ and avoiding $u$. Note that both operations $leg()$ and $belly()$ leave the vertices $u$, $r$ and $t$ unchanged.

\paragraph{Modifying $D$ to $D'$:}
We use the operations $leg()$ and $belly()$ for a preprocessing on the subdivided edges $ab$ and $cd$ (if applicable) by $\Gamma$. Suppose first that $ru \notin \{ab,cd\}$; we will solve the remaining case $ru \in \{ab,cd\}$ later. Assume $birth(ab) \neq birth(b)$ and recall that $birth(a) \leq birth(b)$. If $birth(a) < birth(b)$, $ab$ is a leg of $P_{birth(b)}$ and we apply the operation $leg(a,b)$. Otherwise, $birth(a) = birth(b)$ and we apply the operation $belly(a,b)$. In both cases, this leaves a Mondshein sequence in which $birth(ab) = birth(b)$, i.e.\ $ab$ is contained in the long chain $P_{birth(b)}$.

Similarly, if $birth(cd) \neq birth(d)$, we want to apply either $leg(c,d)$ or $belly(c,d)$ to obtain $birth(cd) = birth(d)$. However, doing this without any restrictions may result in loosing $birth(ab) = birth(b)$, e.g.\ when $cd$ is a belly of $P_{birth(b)}$. Thus, we apply $leg(c,d)$ or $belly(c,d)$ only if $birth(d) < birth(b)$, as then $d$ is no inner vertex of $P_{birth(b)}$. Since $birth(d) \leq birth(b)$, we have therefore $birth(d) \in \{birth(b),birth(cd)\}$. Subdivide the edge $ab$ in $G$ and $P_{birth(ab)}$ with $v$ and likewise subdivide $cd$ with $w$ if applicable for $\Gamma$. Call the resulting sequence $D$; $D$ satisfies $birth(v) = birth(b)$ and $birth(d) \in \{birth(b),birth(w)\}$. We obtain the desired Mondshein sequence $D'$ through $rt'$ and avoiding $u$ from $D$ by distinguishing the following cases (see Figure~\ref{fig:PathReplacement}).

\bigskip
\begin{compactenum}[(1)]
	\item \emph{$\Gamma$ is a vertex-vertex-addition}\\
		Obtain $D'$ from $D$ by adding the new short ear $vw$ to the end of $D$. This way $v$ and $w$ exist when $vw$ is born.
	
	\item \emph{$\Gamma$ is an edge-vertex-addition} \hfill $\triangleright$ $birth(v) = birth(b)$
	\begin{compactenum}[(a)]
		\item $birth(w) > birth(b)$ \hfill $\triangleright$ $w \notin G_{birth(b)}$\\
			Obtain $D'$ from $D$ by adding the new ear $vw$ to the end of $D$. Since $birth(w) > birth(b)$, $v$ has a larger neighbor with respect to $birth$.
		\item $birth(w) < birth(b)$\\
			Then $wv \neq ru'$, as otherwise we would have $w=r$ and $v=u'$ and thus $ab = ru$, which contradicts our assumption.
			Hence, $wv$ is a leg of $P_{birth(v)}$. We apply $leg(w,v)$. By the orientation assigned to $P_{birth(v)}$, this ensures that $v$ has a larger neighbor with respect to $birth$ (e.g., $b$).
		\item $birth(w) = birth(b)$\\
			Then $wv \notin P_{birth(v)}$, since $v$ is adjacent to only $a$ and $b$ in $P_{birth(v)}$ and $w \notin \{a,b\}$ for edge-vertex-additions. Thus, $birth(w) = birth(v) \neq birth(wv)$ and hence $wv$ is a belly of $P_{birth(v)}$. We apply $belly(w,v)$. By the orientation assigned to $P_{birth(v)}$, this ensures that $v$ has a larger neighbor.
	\end{compactenum}

	\item \emph{$\Gamma$ is an edge-edge-addition} \hfill $\triangleright$ $birth(v) = birth(b)$ and $birth(d) \in \{birth(b),birth(w)\}$
	\begin{compactenum}[(a)]
		\item $birth(d) < birth(b)$ \hfill $\triangleright$ $d \in G_{birth(b)-1}$ and $birth(b) > 0$\\
			Then $birth(c) \leq birth(d) = birth(w) < birth(b) = birth(v)$. We further have $vw \neq ru'$, as otherwise we would have $w=r$ and $v=u'$ and thus $r \in \{a,b\}$ which contradicts $r=w$. Hence, $wv$ is a leg of $P_{birth(b)}$. Obtain $D'$ from $D$ by applying $leg(w,v)$.
		\item $birth(d) = birth(b) = birth(w)$ \hfill $\triangleright$ $d,w \in inner(P_{birth(b)})$\\
			Then $vw$ is a belly of $P_{birth(b)}$. Obtain $D'$ from $D$ by applying $belly(v,w)$.
		\item $birth(d) = birth(b) \neq birth(w)$ and $birth(c) = birth(b)$ \hfill $\triangleright$ $c,d \in inner(P_{birth(b)}) \not \ni w$\\
			Then $birth(w) > birth(b)$ and thus $P_{birth(w)} = cw \cup wd$. Let $Z$ be a shortest path in $P_{birth(b)}$ that contains $c$, $d$ and $v$, but not the edge $rt'$ (the latter is only relevant for $birth(b)=0$). Let $z$ be the inner vertex of $Z$ that is contained in $\{c,d,v\}$. At least one of the two paths $Z[;z]$ and $Z[z;]$, say $Z[z;]$, contains an inner vertex, as otherwise $\Gamma$ would not be a BG-operation. Obtain $D'$ from $D$ by deleting $P_{birth(w)}$, replacing the path $Z$ in $P_{birth(b)}$ with the two edges connecting $w$ to the endpoints of $Z$, and adding the two new ears $Z[;z]+w$ and $Z[z;]$ directly afterward in that order. Clearly, $rt' \in P_0$ in $D'$.
		\item $birth(d) = birth(b) \neq birth(w)$ and $birth(c) \neq birth(b)$ \hfill $\triangleright$ $d \in inner(P_{birth(b)}) \not \ni c,w$\\
			Then $birth(c) < birth(d) < birth(w)$ and hence $birth(b) > 0$ and $P_{birth(w)} = cw \cup wd$. One of the paths $P_{birth(b)}[;v]$ and $P_{birth(b)}[v;]$, say $P_{birth(b)}[v;]$, contains $d$ as an inner vertex.
			Obtain $D'$ from $D$ by replacing $P_{birth(b)}$ with the two ears $P_{birth(b)}[;v]+w+c$ and $P_{birth(b)}[v;]$ in that order and replacing $P_{birth(w)}$ with the short ear $wd$. If $birth(b) \neq birth(u)$, it follows directly that $u' = u$ and thus that $D'$ avoids $u' = u$. Otherwise $birth(b)=birth(u)$, which implies $u=b=d$ and $c \neq r$, since we assumed $cd \neq ru$. Thus, in this case $D'$ avoids $u' = u = b$ as well.
	\end{compactenum}
\end{compactenum}

\begin{figure}[htb]
	\captionsetup[subfigure]{labelformat=empty}
	\centering
	\subfloat[Case~(1)]{
		\includegraphics[scale=1.0]{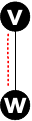}
		\label{fig:Case1}
	}
	\hspace{1.0cm}
	\subfloat[Case~(2a)]{
		\includegraphics[scale=1.0]{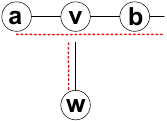}
		\label{fig:Case2ai}
	}
	\hspace{0.5cm}
	\subfloat[Case~(2b)]{
		\includegraphics[scale=1.0]{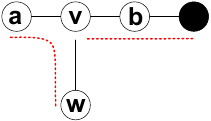}
		\label{fig:Case2aii}
	}
	\hspace{0.5cm}
	\subfloat[Case~(2c)]{
			\includegraphics[scale=1.0]{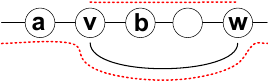}
			\label{fig:Case2aiii}
	}
	\\
	\subfloat[Case~(3a)]{
		\includegraphics[scale=1.0]{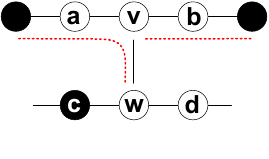}
		\label{fig:Case3aiFirstCase}
	}
	\hspace{0.5cm}
	\subfloat[Case~(3b)]{
		\includegraphics[scale=1.0]{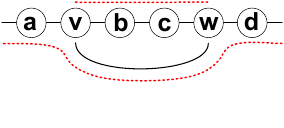}
		\label{fig:Case3aiiThirdCase}
	}
	\\
	\subfloat[Case~(3c)]{
		\includegraphics[scale=1.0]{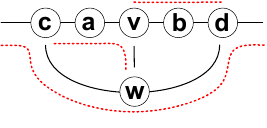}
		\label{fig:Case3aiiSecondCase}
	}
	\hspace{0.5cm}
	\subfloat[Case~(3d)]{
		\includegraphics[scale=1.0]{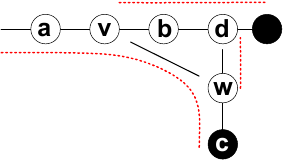}
		\label{fig:Case3aiiFirstCase}
	}
	\caption{Cases when modifying $D$ to $D'$. Black vertices are endpoints of ears that are contained in $G_{birth(b)}$. The dashed paths depict (parts of) the ears in $D'$.}
	\label{fig:PathReplacement}
\end{figure}

In all these cases, we obtain a Mondshein sequence $D'$ through $rt'$ and avoiding $u'$ as desired. Now consider the remaining case $ru \in \{ab,cd\}$. If $birth(d)=birth(b)$ (for an edge-edge-addition), we have $b=d=u$ and can w.l.o.g.\ assume $ru=ab$. Otherwise, $birth(d) < birth(b)$ and it follows directly that we have in all cases, even for edge-vertex-additions, $r=a$ and $u = b$. If $cd$ is a short ear, we move $cd$ to the position in $D$ directly after $P_{birth(d)}$; this preserves a Mondshein sequence. As before, subdivide $ab$ and $cd$ with $v$ and $w$.

Let $\Gamma$ be an edge-vertex-addition. Then $u' = v$ and hence $birth(w) < birth(u) < birth(v)$. Obtain $D'$ from $D$ by replacing $P_{birth(v)}$ with the long ear $uv \cup vw$ and adding the short ear $av = ru'$ directly afterward. Then $D'$ avoids $u'$.

Let $\Gamma$ be an edge-edge-addition and suppose first that $birth(w) \neq birth(u)$. Then $u'=v$ and $birth(w) < birth(v) > birth(u)$. Obtain $D'$ from $D$ by replacing $P_{birth(v)}$ with the long ear $uv \cup vw$ and adding the short ear $av = ru'$ directly afterward. Then $D'$ avoids $u'$. Now suppose that $birth(w) = birth(u)$. Then $b=d=u$, $u'=v$ and $birth(u) = birth(w) < birth(v)$. Obtain $D'$ from $D$ by replacing $P_{birth(v)}$ with the long ear $uv \cup vw$ and adding the short ear $av = ru'$ directly afterward. Hence, in all cases, we obtain a Mondshein sequence $D'$ through $rt'$ and avoiding $u'$.

\paragraph{Computational Complexity:}
For proving the Path Replacement Lemma~\ref{lem:PathReplacement}, it remains to show that each modification can be computed in amortized constant time. Note that ears may become arbitrarily long in the path replacement process and therefore may contain up to $\Theta(n)$ vertices. Moreover, we have to maintain the birth-values of all vertices that are involved in future BG-operations in order to compute which of the subcases in Case~(1)--(3) applies. Thus, we cannot use the standard approach of storing the ears of $D$ explicitly by using doubly-linked lists, as then the birth-values of linearly many vertices may change for every modification.

Instead, we will represent the ears as the sets of a data structure for \emph{set splitting}, which maintains disjoint sets online under an intermixed sequence of find and split operations. Gabow and Tarjan~\cite{Gabow1985} discovered the first data structure for set splitting with linear space and constant amortized time per operation. Their and our model of computation is the standard unit-cost word-RAM. Imai and Asano~\cite{Imai1987} enhanced this data structure to an \emph{incremental variant}, which additionally supports adding single elements to certain sets in constant amortized time. In both results, all sets are restricted to be intervals of some total order. To represent the Mondshein sequence $D$ in the path replacement process, we will use the following more general data structure due to Djidjev~\cite[Section~3.2]{Djidjev2006}, which does not have that requirement but still supports the add-operation.

\smallskip
The data structure maintains a collection $P$ of edge-disjoint paths under the following operations:

\begin{compactitem}
	\item[\texttt{new\_path(x,y)}:] Creates a new path that consists of the edge $xy$. The edge $xy$ must not be in any other path of $P$.
	\item[\texttt{find(e)}:] Returns the integer-label of the path containing the edge $e$.
	\item[\texttt{split(xy)}:] Splits the path containing the edge $xy$ into the two subpaths from $x$ to one endpoint and from $x$ to the other endpoint of that path.
	\item[\texttt{sub(x,e)}:] Modify the path containing $e$ by subdividing $e$ with the vertex $x$.
	\item[\texttt{replace(x,y,e)}:] Neither $x$ nor $y$ may be an endpoint of the path $Z$ containing $e$. Cut $Z$ into the subpath from $x$ to $y$ and the path that consists of the two remaining subpaths of $Z$ joined by the new edge $xy$.
	\item[\texttt{add(x,yz)}:] The vertex $y$ must be an endpoint of the path $Z$ containing the edge $yz$ and $x$ is either a new vertex or not in $Z$. Add the new edge $xy$ to $Z$.
\end{compactitem}
\bigskip

Note that all ears are not only edge-disjoint but also internally disjoint. Djidjev proved that each of the above operations can be computed in amortized constant time~\cite[Theorem~1]{Djidjev2006}. We will only represent long ears in this data structure; the remaining short ears do not contain any essential birth-value information and can therefore be maintained simply as edges. As the data structure can only store paths, we need to clarify how the unique cycle $P_0$ in $D$ can be maintained: We store $P_0$ as paths, namely as the two paths in $P_0$ with endpoints $r$ and $t$. For every ear different from $P_0$, we store its two endpoints at its \texttt{find()}-label. These endpoints can therefore be accessed and updated in constant time.

Now we initialize the data structure with the Mondshein sequence of $K_4$ in constant time using the above operations. Every modification of the Cases~(1)--(3) and $ru \in \{ab,cd\}$ can then be realized with a constant number of operations of the data structure, and hence in amortized constant time.

Additionally, we need to maintain the order of ears in $D$. The \emph{incremental list order-maintenance problem} is to maintain a total order subject to the operations of (i) \emph{inserting} an element after a given element and (ii) \emph{comparing} two distinct given elements by returning the one that is smaller in the order. Bender et al.~\cite{Bender2002} showed a simple solution with amortized constant time per operation (which holds even if, additionally, \emph{deletions} of elements are supported); we will call this the \emph{order data structure}. It is easy to see that the Path Replacement Lemma inserts in every step at most two new ears directly after $P_{birth(b)}$ and at most one new short ear at the end of $D$. Hence, we can maintain the order of ears in $D$ by applying the order data structure to the \texttt{find()}-labels of ears; this costs amortized constant time per step.

For deciding which of the subcases in~(1)--(3) and $ru \in \{ab,cd\}$ applies, we additionally need to maintain the birth-values of the vertices and edges in $D$. In fact, it suffices to support the queries ``$birth(x) < birth(y)$'' and ``$birth(x) = birth(y)$'', where $x$ and $y$ may be arbitrary edges or vertices in $D$.
If $x$ and $y$ are edges, both queries can be computed in constant amortized time by comparing the labels \texttt{find(x)} and \texttt{find(y)} in the order data structure. In order to allow birth-queries on vertices, we will store pointers at every vertex $x$ to the two edges $e_1$ and $e_2$ that are incident to $x$ in $P_{birth(x)}$. The desired query involving $birth(x)$ can then be computed by comparing \texttt{find(e$_1$)} in the order data structure.

For any new vertex $x$ that is added to $D$, we can find $e_1$ and $e_2$ in constant time, as these are in $\{av,vb,cw,wd,vw\}$. Since $P_{birth(x)}$ may change over time, we have to update $e_1$ and $e_2$ after each step. The only situation in which $P_{birth(x)}$ may loose $e_1$ or $e_2$ (but not both) is a \texttt{split} or \texttt{replace} operation on $P_{birth(x)}$ at $x$ (the split operation must be followed by an add operation on $x$, as $x$ is always inner vertex of some ear). This cuts $P_{birth(x)}$ into two paths, each of which contains exactly one edge in $\{e_1,e_2\}$. Checking \texttt{find(e$_1$)$=$find(e$_2$)} recognizes this case efficiently. Dependent on the particular case, we compute a new consistent pair $\{e'_1,e'_2\}$ that differs from $\{e_1,e_2\}$ in exactly one edge. This allows to check the desired comparisons in amortized constant time.

We conclude that $D'$ can be computed from $D$ in amortized constant time; this proves the Path Replacement Lemma. Thus, we deduce the following theorem.

\begin{theorem}\label{thm:Linear}
Given edges $rt$ and $ru$ of a $3$-connected graph $G$, a Mondshein sequence $D$ of $G$ through $rt$ and avoiding $u$ can be computed in time $O(m)$.
\end{theorem}

The above algorithm is \emph{certifying} in the sense of~\cite{McConnell2011}: First, check in linear time that $D$ is an ear decomposition of $G$. Second, check the side constraints on the first and last ear. Third, check in linear time that $D$ is non-separating by testing that every ear satisfies Definition~\ref{def:nonseparating}.

\section{Applications}\label{sec:applications}

\paragraph{Application 1:} Independent Spanning Trees\\
Let $k$ spanning trees of a graph be \emph{independent} if they all have the same root vertex $r$ and, for every vertex $x \neq r$, the paths from $x$ to $r$ in the $k$ spanning trees are \emph{internally disjoint} (i.e., vertex-disjoint except for their endpoints; see Figure~\ref{fig:Cr01ISTGrey}). The following conjecture from 1988 due to Itai and Rodeh~\cite{Itai1988} has received considerable attention in graph theory throughout the past decades.

\begin{conjecture}[Independent Spanning Tree Conjecture~\cite{Itai1988}]
Every $k$-connected graph contains $k$ independent spanning trees.
\end{conjecture}

\begin{figure}[h!tb]
	\centering
	\includegraphics[scale=0.8]{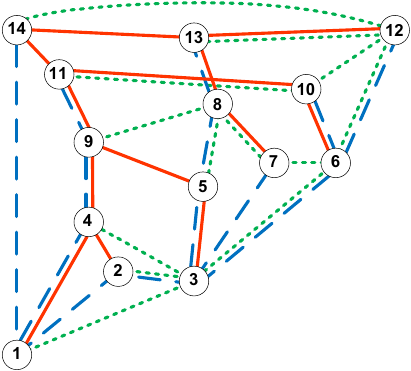}
	\caption{Three independent spanning trees in the graph of Figure~\ref{fig:SequenceCr08Paper}, which were computed from its Mondshein sequence (vertex numbers depict a consistent $tr$-numbering).}
	\label{fig:Cr01ISTGrey}
\end{figure}

The conjecture has been proven for $k \leq 2$~\cite{Itai1988}, $k=3$~\cite{Cheriyan1988,Zehavi1989} and $k=4$~\cite{Curran2006}, with running times $O(m)$, $O(n^2)$ and $O(n^3)$, respectively, for computing the corresponding independent spanning trees. For every $k \geq 5$, the conjecture is open. For planar graphs, the conjecture has been proven by Huck~\cite{Huck1999}.

We show how to compute three independent spanning trees in linear time, using an idea of~\cite{Cheriyan1988}. This improves the previous best quadratic running time. It may seem tempting to compute the spanning trees directly and without using a Mondshein sequence, e.g.\ by local replacements in an induction over BG-operations or inverse contractions. However, without additional restrictions this is bound to fail, as shown in Figure~\ref{fig:ISTCounterexample}.

\begin{figure}[h!tb]
	\centering
	\includegraphics[scale=0.8]{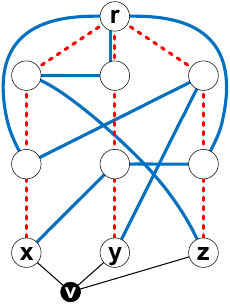}
	\caption{A $3$-connected graph $G$ (some edges are not drawn). $G$ is obtained from the $3$-connected graph $G':= (G-v) \cup xy$ by performing a BG-operation (or inverse contraction) that adds the vertex $v$ (with added edge $vy$). Two of the three independent spanning trees of $G'$ are given, rooted at $r$ (thick edges). However, not both of them can be extended to $v$.}
	\label{fig:ISTCounterexample}
\end{figure}

Compute a Mondshein sequence through $rt$ and avoiding $u$, as described in Theorem~\ref{thm:Linear}. Choose $r$ as the common root vertex of the three spanning trees and let $x \neq r$ be an arbitrary vertex.

First, we show how to obtain two internally disjoint paths from $x$ to $r$ that are both contained in the subgraph $G_{birth(x)}$. A \emph{$tr$-numbering} $<$ is a total order $v_1 < \cdots < v_n$ of the vertices of a graph such that $t = v_1$, $r = v_n$, and every other vertex has both a higher-numbered and a lower-numbered neighbor. Let a $tr$-numbering $<$ be \emph{consistent}~\cite{Cheriyan1988} to a Mondshein sequence if $<$ is a $tr$-numbering for every graph $G_i$, $0 \leq i \leq m-n$. We can compute a consistent $tr$-numbering $<$ in linear time as follows: Let $<_0$ be the total order on $V(P_0)$ from $t$ to $r$; then $<_0$ is a consistent $tr$-numbering of $G_0$. We maintain $<_{i-1}$ in the \emph{order data structure} of~\cite{Bender2002} (see the computational complexity paragraph). Now we add iteratively the next ear $P_i$ and obtain $<_i$ from $<_{i-1}$ by ordering the new inner vertices of $P_i$ from the lower to the larger endpoint of $P_i$ in $<_{i-1}$ (such that $inner(P_i)$ is between these endpoints in $<_i$). This takes amortized time proportional to the length of $P_i$ and, hence, gives a total linear running time.

According to $<$, every vertex $x \neq r$ has a higher-numbered neighbor in $G_{birth(x)}$ and every vertex $x \notin \{r,t\}$ a lower-numbered neighbor in $G_{birth(x)}$. Fixing arbitrary such neighbors, the first two spanning trees $T_1$ and $T_2$ then consist of the incident edges to higher neighbors and of the edge $tr$ and the incident edges to lower neighbors, respectively. Clearly, $T_1$ and $T_2$ are independent due to the numbering used.

We construct the third independent spanning tree $T_3$. As a Mondshein sequence is non-separating, every vertex $x \neq \{r,u\}$ has an incident edge with an endpoint in $\overline{G_{birth(x)}}$ (as seen before, iterating this argument gives a path to $u$ in $\overline{G_{birth(x)}}$). Let $T_3$ consist of arbitrary such incident edges and of the edge $ru$. Since $G_{birth(x)}$ and $\overline{G_{birth(x)}}$ are vertex-disjoint, $T_3$ is independent from $T_1$ and $T_2$.

\begin{remark}
We remark that the three independent spanning trees constructed this way satisfy the following additional condition: Due to the fact that $T_2$ and $T_3$ are extended to $r$ by one single edge, all incident edges of $r$ are contained in at most one of $T_1,T_2,T_3$. In particular, no edge of $G$ is contained in all three independent trees, which is a fact that cannot be derived from the definition of independent spanning trees (an edge that is incident to $r$ may be contained in all three trees).
\end{remark}

\paragraph{Application 2:} Output-Sensitive Reporting of Disjoint Paths\\
Given two vertices $x$ and $y$ of an arbitrary graph, a \emph{$k$-path query} reports $k$ internally disjoint paths between $x$ and $y$ or outputs that these do not exist. Di Battista, Tamassia and Vismara~\cite{DiBattista1999} give data structures that answer $k$-path queries for $k \leq 3$. A key feature of these data structures is that every $k$-path query has an \emph{output-sensitive} running time, i.e., a running time of $O(\ell)$ if the total length of the reported paths is $\ell$ (and running time $O(1)$ if the paths do not exist). The preprocessing time of these data structures is $O(m)$ for $k \leq 2$, but $O(n^2)$ for $k=3$.

For $k=3$, Di Battista et al.\ show how the input graph can be restricted to be $3$-connected using a standard decomposition. For every $3$-connected graph we can compute a Mondshein sequence, which allows us to compute three independent spanning trees $T_1$--$T_3$ in a linear preprocessing time, as shown in Application~1. If $x$ or $y$ is the root $r$ of $T_1$--$T_3$, this gives a straight-forward output-sensitive data structure that answers $3$-path queries: we just store $T_1$--$T_3$ and extract one path from each tree per query.

In order to extend these queries to $k$-path queries between arbitrary vertices $x$ and $y$, \cite{DiBattista1999} gives a case distinction that shows that the desired paths can be found efficiently in the union of the six paths in $T_1$--$T_3$ that join either $x$ with $r$ or $y$ with $r$. This case distinction can be used for the desired output-sensitive reporting in time $O(\ell)$ without changing the preprocessing. We conclude that the preprocessing time of $O(n^2)$ for allowing $k$-path queries with $k \leq 3$ in arbitrary graphs can be improved to $O(n+m)$.

\paragraph{Application 3:} Planarity Testing\\
We give a conceptually very simple planarity test based on Mondshein's sequence for any $3$-connected graph $G$ in time $O(n)$. The $3$-connectivity requirement is not crucial, as the planarity of $G$ can be reduced to the planarity of all $3$-connected components of $G$, which in turn are computed as a side-product from the computation of the BG-sequence~\cite[Appendix~2]{Mehlhorn2015}. Alternatively, one could also use standard algorithms~\cite{Hopcroft1973,Gutwenger2001} for reducing $G$ to be $3$-connected.

If $m > 3n-6$, $G$ is not planar due to Euler's formula and we reject the instance, so let $m \leq 3n-6$.
Let $rt$ be an edge of $G$. We will find an embedding whose outer face is left of $rt$, unless $G$ is non-planar. Due to Whitney~\cite{Whitney1932}, this embedding is unique. In light of Observation~\ref{obs:MondsheinIsCanonical}, we need to pick an edge $ru \neq rt$ such that $tr \cup ru$ is in a non-separating cycle. We can easily find such an edge by computing a Mondshein sequence through $rt$ and avoiding some vertex $u' \notin \{r,t\}$, and then taking the edge that is incident to $r$ in $P_0-rt$ (alternatively, any linear-time algorithm that computes a non-separating cycle containing $rt$ like the one in~\cite{Cheriyan1988} can be used).

Now we compute a Mondshein sequence $D$ through $rt$ and avoiding $u$ that satisfies Property~\ref{lem:canonical}.7 in time $O(n)$. If $G$ is planar, Observation~\ref{obs:MondsheinIsCanonical} ensures that $D$ is a canonical ordering of our fixed embedding; in particular, the last vertex $u$ and the edge $rt$ will be embedded in the outer face. Due to Property~\ref{lem:canonical}.7, $P_0$ has no chords and every short ear $xy$ satisfies $birth(x) \neq birth(y)$. For the embedding process, we rearrange the order of short ears in $D$ such that all short ears $xy$ with $birth(x) < birth(y)$ are direct successors of the long ear $P_{birth(y)}$ (this can be done in linear time using bucket sort).

We start with a planar embedding $M_0$ of $P_0$. Step by step, we attempt to augment $M_i$ with the next long ear $P_j$ in $D$ as well as all short ears directly succeeding $P_j$ in order to construct a planar embedding $M_j$ of $G_j$.

Once the current embedding $M_i$ contains $u$, we have added all edges of $G$ and are done. Otherwise, $u$ is contained in $\overline{G_i}$, according to Definition~\ref{def:canonicalordering}.2. Then $\overline{G_i}$ contains a path from each inner vertex of $P_j$ to $u$, according to Lemma~\ref{lem:connected}. Since $u$ is contained in the outer face of the unique embedding of $G$, adding the long ear $P_j$ to $M_i$ can preserve planarity only when it is embedded into the outer face $f$ of $M_i$. Thus, we only have to check that both endpoints of $P_j$ are contained in $f$ (this is easy to test by maintaining the vertices of the outer face). For the same reason, the short ears directly succeeding $P_j$ can preserve planarity only if the set $S$ of their endpoints in $G_i$ is contained in $f$. Note that, if there is at least one such short ear, $P_j$ has precisely one inner vertex $v$ due to Property~\ref{lem:canonical}.7 and all short ears directly succeeding $P_j$ have $v$ as endpoint.

Thus, if the endpoints of $P_j$ and $S$ are contained in $f$, we embed $P_j$ and the short ears into $f$ in the only possible way, i.e.\ as a path or as one new vertex $v$ with the short ears and the two edges of $P_j$ as incident edges. Otherwise, we output ``not planar''. If desired, a Kuratowski-subdivision can then be easily extracted in time $O(n)$, as shown in~\cite[Lemma~5]{Schmidt2013b} (the extraction is even simpler, as we do not make use of adding ``claws'').

\paragraph{Application 4:} Contractible Subgraphs in 3-Connected Graphs\\
A connected subgraph $H$ of a $3$-connected graph $G$ is called \emph{contractible} if contracting $H$ to a single vertex generates a $3$-connected graph. It is easy to show that a connected subgraph $H$ is contractible if and only if $G-V(H)$ is $2$-connected. While many structural results about contractible subgraphs are known in graph theory, we are not aware of any non-trivial result that computes them.

Using a Mondshein sequence, we can identify a nested family of $m-n$ contractible induced subgraphs in linear time, namely the subgraphs $\overline{G_i}$ for every $0 \leq i < m-n$. Clearly, these subgraphs are contractible, as $G - \overline{G_i}$ is $2$-connected due to Lemma~\ref{lem:canonical}.5. Moreover, for each $i > 0$, $\overline{G_i}$ is an induced subgraph of the induced subgraph $\overline{G_{i-1}}$. In particular, every $\overline{G_i}$ contains $u$, since $\overline{V_{m-n-1}} = \{u\}$ due to Definition~\ref{def:Mondshein}.2.

\paragraph{Application 5:} The \emph{k-Partitioning} Problem\\
Given vertices $a_1,\ldots,a_k$ of a graph $G$ and natural numbers $n_1,\ldots,n_k$ with $n_1+\cdots+n_k=n$, we want to find a partition of $V$ into sets $A_1,\ldots,A_k$ with $a_i \in A_i$ and $|A_i| = n_i$ for every $i$ such that every set $A_i$ induces a connected graph in $G$. We call this a \emph{$k$-partition}.

If the conditions $a_i \in A_i$ are ignored, the problem becomes NP-hard even for $k=2$ and bipartite input graph $G$~\cite{Dyer1985}; although often stated otherwise, this does not seem to imply an NP-hardness proof for the $k$-partitioning problem directly. If the input graph is $k$-connected, however, Gy{\"o}ri~\cite{Gyori1981} and Lov\'asz~\cite{Lovasz1977} proved that there is always a $k$-partition. Thus, let $G$ be $k$-connected. If $k=2$, the $k$-partitioning problem is easy to solve: If $G$ does not contain the edge $a_1a_2$, add this edge to $G$. Compute an $a_1a_2$-numbering $a_1=v_1,v_2,\ldots,v_n=a_2$ and observe that, for any vertex $v_i$ (in particular for $v_{n_1}$), the graphs induced by $\{v_1,\ldots,v_i\}$ and by $\{v_{i+1},\ldots,v_n\}$ are connected. For every $k \geq 4$, the $k$-partitioning problem on a $k$-connected input graph is not even known to be in $P$ (although its decision variant is), so we will focus on the $3$-partitioning problem of a $3$-connected input graph.

This problem can be solved in quadratic time~\cite{Suzuki1990} and, if the graph is additionally \emph{planar}, even in linear time~\cite{Jou1994}. As suggested in~\cite{Wada1993,Awal2010}, the problem (as well as a related extension) can be solved with the aid of a non-separating ear decomposition. For planar graphs, it thus suffices with Observation~\ref{obs:MondsheinIsCanonical} to compute just a canonical ordering, which simplifies previous algorithms considerably.

More generally, we get the first $O(m)$ time algorithm for arbitrary $3$-connected graphs as follows. Consider a Mondshein sequence through $a_1a_2$ and avoiding $a_3$ (if the edges $a_1a_2$ and $a_1a_3$ do not exist in $G$, we add them in advance). If $G_i$ contains exactly $n_1+n_2$ vertices for some $i$, we set $A_3 := \overline{G_i}$ and compute $A_1$ and $A_2$ by solving the $2$-partitioning problem on $G_i$ in linear time using an $a_1a_2$-numbering, as described above. Otherwise, let $P_i$ be the first ear such that $|V(G_i)| > n_1+n_2$.

We partition $inner(P_i)$ into the vertex sets $B_1$, $B_3$ and $B_2$ (designated to be part of $A_1$, $A_3$ and $A_2$, respectively) of three consecutive paths in $P_i-a_1a_2$ such that $|B_3|=n_3-|\overline{V_i}|$. In particular, $0 < |B_3| < |inner(P_i)|$. Let $l := |B_1|+|B_2|$; then there are $l+1$ choices for $B_3$. For any such choice, setting $A_3 := B_3 \cup \overline{V_i}$ satisfies the claim for $A_3$, as $A_3$ contains $a_3$, has cardinality $n_3$ and is connected, as a Mondshein sequence is non-separating.

We specify how to compute $B_1$; this determines the sets $B_3$ and $B_2$. If $i=0$, choose $B_1$ as the path in $P_0-a_1a_2$ that starts at $a_1$ and consists of $n_1$ vertices. The desired $2$-partition of $G-A_3$ is then given by $A_1 := B_1$ and $A_2 := B_2$. If $i > 0$, we aim for a coloring of $G_{i-1}$ into \emph{blue} and \emph{red} vertices such that $A_1$ consists of $B_1$ and the blue vertices, and $A_2$ consists of $B_2$ and the red vertices. In order to make $A_1$ connected, we have to prevent that both endpoints of $G_{i-1}$ are colored red as long as $|B_1| > 0$. Clearly, $|B_1| < n_1$, as $a_1$ has to be in $A_1$; similarly, $|B_2| < n_2$, which implies $|B_1| > l-n_2$. Hence, the valid choices for $|B_1|$ are between $max\{0,l-n_2+1\}$ and $min\{l,n_1-1\}$.

For every $max\{0,l-n_2+1\} \leq |B_1| \leq min\{l,n_1-1\}$, we compute a 2-partition of $G_{i-1}$ into $n_1-|B_1|$ blue and $n_2-|B_2|$ red vertices. The first 2-partition for $|B_1| = max\{0,l-n_2+1\}$ can be computed in linear time using an $a_1a_2$-numbering as described above. For each increase of $|B_1|$ by one, we can construct the new 2-partition in constant time from the old one, as exactly one blue vertex is recolored red. If the coloring of one of these choices for $|B_1|$ colors the endpoints $x$ and $y$ of $P_i$ differently, we choose $B_1$ as the path in $P_i$ next to the blue endpoint that consists of $|B_1|$ vertices. Then $A_1$ and $A_2$ as stated above give a 3-partition.

Otherwise, $x$ and $y$ have always the same color. Moreover, this color is identical, say red by symmetry, for every computed choice of $|B_1|$, since only one vertex is recolored per increase of $|B_1|$. Consider the smallest choice $|B_1| := max\{0,l-n_2+1\}$. As $x$ and $y$ are red, $n_2 - |B_2| \geq 2$, which implies $|B_1| > l-n_2+1$. Hence, $|B_1|=0$ and we choose $B_1 := \emptyset$. Then $A_1$ and $A_2$ as stated above give the desired 3-partition.

\paragraph{Acknowledgments.} I wish to thank Joseph Cheriyan for valuable hints, the anonymous person who drew my attention to Lee F.\ Mondshein's work, David R.\ Wood for suggesting the graph in Figure~\ref{fig:ISTCounterexample}, and the anonymous reviewers that gave me very valuable feedback, which led to a reduction of the number of cases in the Path Replacement Lemma.

\bibliographystyle{abbrv}
\bibliography{../../Jens}

\end{document}